\documentclass[reqno]{mcom-l}

\usepackage{amssymb}

\usepackage{url}
\usepackage{enumerate}

\newtheorem{thm}{Theorem}
\newtheorem{prop}[thm]{Proposition}
\newtheorem{cor}[thm]{Corollary}
\newtheorem*{hyp}{Hypothesis P}

\newcounter{stepcounter}
\newcommand{\step}[1]{\refstepcounter{stepcounter} {\sl Step (\roman{stepcounter}) --- #1}}

\newcommand{\ZZ}{\mathbb{Z}}
\newcommand{\RR}{\mathbb{R}}
\newcommand{\NN}{\mathbb{N}}
\newcommand{\FF}{\mathbb{F}}
\newcommand{\Rring}{\mathcal{R}}
\newcommand{\Sring}{\mathcal{S}}
\newcommand{\ii}{\mathbf{i}}
\newcommand{\divides}{\mathrel|}
\newcommand{\eps}{\varepsilon}
\DeclareMathOperator{\Mint}{\mathsf{M}}
\DeclareMathOperator{\Tpoly}{\mathsf{T}}
\DeclareMathOperator{\Tshort}{\mathsf{T}_{\mathrm{short}}}
\DeclareMathOperator{\Cshort}{\mathsf{C}_{\mathrm{short}}}
\DeclareMathOperator{\Cbiv}{\mathsf{C}_{\mathrm{biv}}}
\DeclareMathOperator{\Ctiny}{\mathsf{C}_{\mathrm{tiny}}}
\renewcommand{\leq}{\leqslant}
\renewcommand{\geq}{\geqslant}

\begin{document}

\title[Faster integer multiplication]{Faster integer multiplication using \\ plain vanilla FFT primes}


\author{David Harvey}
\address{School of Mathematics and Statistics, University of New South Wales, Sydney NSW 2052, Australia}
\curraddr{}
\email{d.harvey@unsw.edu.au}
\thanks{Harvey was supported by the Australian Research Council (grants DP150101689 and FT160100219).}

\author{Joris van der Hoeven}
\address{Laboratoire d'informatique, UMR 7161 CNRS, \'Ecole polytechnique, 91128 Palaiseau, France}
\curraddr{}
\email{vdhoeven@lix.polytechnique.fr}
\thanks{}


\date{}

\dedicatory{}

\begin{abstract}
Assuming a conjectural upper bound for the least prime in an arithmetic progression, we show that $n$-bit integers may be multiplied in $O(n \log n\, 4^{\log^* n})$ bit operations.
\end{abstract}

\maketitle


\bibliographystyle{amsplain}

\section{Introduction}

Let $\Mint(n)$ be the number of bit operations required to multiply two $n$-bit integers in the deterministic multitape Turing model \cite{Pap-complexity}.
A decade ago, F\"urer \cite{Fur-faster1,Fur-faster2} proved that
\begin{equation}
\label{eq:furer}
 \Mint(n) = O(n \log n \, K^{\log^* n})
\end{equation}
for some constant $K > 1$.
Here $\log^* x$ denotes the iterated logarithm, that is,
 \[ \log^* x := \min\{ j \in \NN : \log^{\circ j} x \leq 1\} \qquad (x \in \RR), \]
where $\log^{\circ j} x := \log \cdots \log x$ (iterated~$j$ times).
Harvey, van der Hoeven and Lecerf~\cite{HvdHL-mul} subsequently gave a related algorithm that achieves \eqref{eq:furer} with the explicit value $K = 8$, and more recently Harvey announced that one may achieve $K = 6$ by applying new techniques for truncated integer multiplication \cite{Har-truncmul}.

There have been two proposals in the literature for algorithms that achieve the tighter bound
\begin{equation}
\label{eq:K=4}
 \Mint(n) = O(n \log n \, 4^{\log^* n})
\end{equation}
under plausible but unproved number-theoretic hypotheses.
First, Harvey, van der Hoeven and Lecerf gave such an algorithm \cite[\S9]{HvdHL-mul} that depends on a slight weakening of the Lenstra--Pomerance--Wagstaff conjecture on the density of Mersenne primes, that is, primes of the form $p = 2^m - 1$, where~$m$ is itself prime.
Although this conjecture is backed by reasonable heuristics and some numerical evidence, it is problematic for several reasons.
At the time of writing, only 49 Mersenne primes are known, the largest being $2^{74{,}207{,}281}-1$ \cite{M74207281}.
More significantly, it has not been established that there are infinitely many Mersenne primes.
Such a statement seems to be well out of reach of contemporary number-theoretic methods.

A second conditional proof of \eqref{eq:K=4} was given by Covanov and Thom\'e \cite{CT-zmult}, this time assuming a conjecture on the density of certain generalised Fermat primes, namely, primes of the form $r^{2^\lambda} + 1$.
Again, although their unproved hypothesis is supported by heuristics and some numerical evidence, it is still unknown whether there are infinitely many primes of the desired form.
It is a famous unsolved problem even to prove that there are infinitely many primes of the form $n^2 + 1$, of which the above generalised Fermat primes are a special case.

As an aside, we mention that the unproved hypotheses in \cite[\S9]{HvdHL-mul} and \cite{CT-zmult} may both be expressed as statements about the cyclotomic polynomials $\phi_k(x)$ occasionally taking prime values: for \cite[\S9]{HvdHL-mul} we have $2^m - 1 = \phi_m(2)$, and for \cite{CT-zmult} we have $r^{2^\lambda} + 1 = \phi_{2^{\lambda+1}}(r)$.

In this paper we give a new conditional proof of \eqref{eq:K=4}, which depends on the following hypothesis.
Let $\varphi(q)$ denote the totient function.
For relatively prime positive integers~$r$ and~$q$, let $P(r, q)$ denote the least prime in the arithmetic progression $n = r \pmod q$, and put $P(q) := \max_r P(r, q)$.
\begin{hyp}
We have $P(q) = O(\varphi(q) \log^2 q)$ as $q \to \infty$.
\end{hyp}
Our main result is as follows.
\begin{thm}
\label{thm:main}
Assume Hypothesis P.
Then there is an algorithm achieving \eqref{eq:K=4}.
\end{thm}

The overall structure of the new algorithm is largely inherited from \cite[\S9]{HvdHL-mul}.
In particular, we retain the strategy of reducing a ``long'' DFT (discrete Fourier transform) to many ``short'' DFTs via the Cooley--Tukey method \cite{CT-fft}, and then converting these back to convolution problems via Bluestein's trick \cite{Blu-dft}.
The main difference between the new algorithm and \cite[\S9]{HvdHL-mul} is the choice of coefficient ring.
The algorithm of \cite[\S9]{HvdHL-mul} worked over $\FF_p[\ii]$, where $p = 2^m - 1$ is a Mersenne prime and $\ii^2 = -1$.
In the new algorithm we use instead the ring $\FF_p$, where~$p$ is a prime of the form $p = a \cdot 2^m + 1$, for an appropriate choice of $m$ and $a = O(m^2)$.
Hypothesis~P guarantees that such primes exist (take $q = 2^m$ and $r = 1$).
The key new ingredient is the observation that we may convert an integer product modulo $a \cdot 2^m + 1$ to a polynomial product modulo $X^k + a$, by splitting the integers into chunks of $m/k$ bits.

In software implementations of fast Fourier transforms (FFTs) over finite fields, such as Shoup's NTL library \cite{ntl-9.9.1}, it is quite common to work over $\FF_p$ where $p$ is a prime of the form $a \cdot 2^m + 1$ that fits into a single machine register.
Such primes are sometimes called \emph{FFT primes}; they are popular because it is possible to perform a radix-two FFT efficiently over $\FF_p$ with a large power-of-two transform length.
Our Theorem \ref{thm:main} shows that such primes remain useful even in a theoretical sense as $m \to \infty$.

From a technical point of view, the new algorithm is considerably simpler than that of \cite[\S9]{HvdHL-mul}.
The main reason for this is that we have complete freedom in our choice of~$m$ (the coefficient size), and in particular we may easily choose $m$ to be divisible by the desired chunk size.
By contrast, the choice of $m$ in \cite[\S9]{HvdHL-mul} is dictated by the rather erratic distribution of Mersenne primes; this forces one to deal with the technical complication of splitting integers into chunks of `non-integral size', which was handled in \cite[\S9]{HvdHL-mul} by adapting an idea of Crandall and Fagin \cite{CF-DWT}.

We remark that the conditional algorithm of Covanov and Thom\'e \cite{CT-zmult} achieves $K = 4$ by a rather different route.
Instead of using Bluestein's trick to handle the short transforms, the authors follow F\"urer's original strategy, which involves constructing a coefficient ring containing ``fast'' roots of unity.
The algorithm of this paper, like the algorithms of \cite{HvdHL-mul}, makes no use of such ``fast'' roots.

Let us briefly discuss the evidence in favour of Hypothesis P.
The best unconditional bound for $P(q)$ is currently Xylouris's refinement of Linnik's theorem, namely $P(q) = O(q^{5.18})$ \cite{Xyl-linnik}.
If $q$ is a prime power (the case of interest in this paper), one can obtain $P(q) = O(q^{2.4+\eps})$ \cite[Cor.~11]{Cha-charsums}.
Assuming the Generalised Riemann Hypothesis (GRH), one has $P(q) = O(q^{2+\eps})$ \cite{HB-linnik}.
All of these bounds are far too weak for our purposes.

The tighter bound in Hypothesis P was suggested by Heath--Brown \cite{HB-almost,HB-siegel}.
It can be derived from the reasonable assumption that a randomly chosen integer in a given congruence class should be no more or less `likely' to be prime than a random integer of the same size, after correcting the probabilities to take into account the divisors of $q$.
A detailed discussion of this argument is given by Wagstaff \cite{Wag-least}, who also presents some supporting numerical evidence.
In the other direction, Granville and Pomerance \cite{GP-least} have conjectured that $\varphi(q) \log^2 q = O(P(q))$.
These questions have been revisited in a recent preprint of Li, Pratt and Shakan~\cite{LPS-leastprime}; they give further numerical data, and propose the more precise conjecture that
 \[ \liminf_q \frac{P(q)}{\varphi(q) \log^2 q} = 1, \qquad \limsup_q \frac{P(q)}{\varphi(q) \log^2 q} = 2. \]
The consensus thus seems to be that $\varphi(q) \log^2 q$ is the right order of magnitude for $P(q)$, although a proof is apparently still elusive.

For the purposes of this paper, there are several reasons that Hypothesis P is much more compelling than the conjectures required by \cite[\S9]{HvdHL-mul} and \cite{CT-zmult}.
First, it is well known that there are infinitely many primes in any given congruence class, and we even know that asymptotically the primes are equidistributed among the congruence classes modulo $q$.
Second, one finds that, in practice, primes of the required type are extremely common.
For example, we find that $a \cdot 2^{1000} + 1$ is prime for
 \[ a = 13, 306, 726, 2647, 3432, 5682, 5800, 5916, 6532, 7737, 8418, 8913, 9072, \ldots \]
and there are still plenty of opportunities to hit primes before exhausting the possible values of $a$ up to about $10^6$ allowed by Hypothesis P.

Third, we point out that Hypothesis P is actually much \emph{stronger} than what is needed in this paper.
We could prove Theorem \ref{thm:main} assuming only the weaker statement that there exists a logarithmically slow function $\Phi(q)$ (see \cite[\S5]{HvdHL-mul}) such that $P(q) < \varphi(q) \Phi(q)$ for all large $q$.
For example, we could replace $(\log q)^2$ in Hypothesis~P by $(\log q)^C$ for any fixed $C > 2$, or even by $(\log q)^{(\log \log q)^C}$ for any fixed $C > 0$.
To keep the complexity arguments in this paper as simple as possible, we will only give the proof of Theorem \ref{thm:main} for the simplest form of Hypothesis P, as stated above.

It is interesting to ask for what bit size the new algorithm would be
faster than the asymptotically inferior multiplication algorithms that
are used in practice.  The current reference implementation in
the GMP library~\cite{gmp} uses Sch\"onhage--Strassen's
algorithm~\cite{SS-multiply}. On recent architectures,
Pollard's algorithm~\cite{Pol71} tends to be very competitive
as well~\cite{Har14}.  Concerning our new algorithm,
we stress that the present paper is optimised for
simplicity rather than speed.  An optimised version would
essentially coincide with Pollard's algorithm for sizes $n<2^{64}$,
where $64$ corresponds to the bit size of an integer hardware register,
so the main recursive step in our algorithm would
only be invoked for super-astronomical sizes.
This does not withstand that some of the techniques that
we developed for proving the new complexity bound
may have practical applications for much smaller sizes.
This is exactly what happened for the related problem of
carryless integer multiplication: new techniques
from~\cite{vdH:ffmul} led to faster practical
implementations~\cite{vdH:f2kmul,vdH:ff2mul}.

\section{The algorithm}

Define $\lg x := \lceil \log_2 x \rceil$ for $x \geq 1$.
\emph{For the rest of the paper we assume that Hypothesis~P holds}, and hence we may fix an absolute constant $C > 0$ such that
 \[ P(q) \leq C q  (\lg q)^2 \]
for all $q \geq 2$.
(Numerical evidence suggests that $P(q)/(q(\lg q)^2)$ achieves its maximum value at $q = 2$, implying that one may take $C = 3/2$.)

An \emph{admissible size} is an integer $m > 2^{17}$ of the form
 \[ m = k  (\lg k)^3 \]
for some integer $k$.
For such $m$, let $p_0(m)$ denote the smallest prime of the form
\[ p = a \cdot 2^m + 1. \]
Hypothesis P implies that
\begin{equation}
\label{eq:a-bound}
 1 \leq a < C m^2.
\end{equation}

In the proof of Proposition \ref{prop:main-bound} below, we will describe a recursive algorithm $\textsc{Transform}$ that takes as input an admissible size $m = k(\lg k)^3$, the corresponding prime $p = p_0(m) = a \cdot 2^m + 1$, a power-of-two transform length $L = 2^{\lg L}$ such that
\begin{equation}
\label{eq:lgL-interval}
 (\lg m)^4 < \lg L < m,
\end{equation}
a primitive $L$-th root of unity $\zeta \in \FF_p$ (such a primitive root exists as $\lg L < m$ and $2^m \divides p - 1$), and a polynomial $F \in \FF_p[X]/(X^L - 1)$.
Its output is the DFT of $F$ with respect to $\zeta$, that is, the vector
 \[ \hat F := (F(1), F(\zeta), \ldots, F(\zeta^{L-1})) \in (\FF_p)^L. \]
For sufficiently large $m$, Proposition \ref{prop:main-bound} will reduce the computation
of such a DFT to the computation of a large collection of similar DFTs
over $\FF_{p'}$ for an exponentially smaller prime
\[ p' = p_0(m') = a' \cdot 2^{m'} + 1 \]
where
\[ m' = k' (\lg k')^3 \sim 2 (\lg m)^3. \]
The main reduction consists of five steps,
which may be summarised as follows:
   \begin{enumerate}[(i)]
      \item Reduce the given `long' transform of length~$L<2^m$ over $\FF_p$
               to many `short' transforms of exponentially smaller
               length $S=2^{(\lg m)^2}$ over $\FF_p$, via the Cooley--Tukey decomposition.
      \item Reduce each short transform to a product in $\FF_p[X]/(X^S - 1)$,
               i.e., a cyclic convolution of length $S$, using Bluestein's algorithm.
      \item By splitting each coefficient in $\FF_p$ into exponentially smaller chunks
               of bit size $r=(\lg k)^3 \sim (\lg m)^3$, reduce each product from step (ii)
               to a product in $\ZZ[X,Y]/(X^S - 1, Y^k + a)$.
      \item Embed each product from step (iii) into $\FF_{p'}[X,Y]/(X^S - 1, Y^k + a)$,
               for a suitable prime $p' = a' \cdot 2^{m'} + 1$
               that is exponentially smaller than~$p$;
               more precisely, $m' \sim 2 (\lg m)^3 \sim 2 r$ .
      \item Reduce each product from step (iv) to a collection of
               forward and inverse DFTs of length~$S$ over $\FF_{p'}$, and recurse.
   \end{enumerate}

\subsection{The main recursion}

Let us now present the main reduction in more detail together
with its complexity analysis.  We denote the running time of
$\textsc{Transform}$ by $\Tpoly(m, L)$. For $m > 2^{17}$
there is always at least one integer $\lg L$ in the interval
\eqref{eq:lgL-interval}, so we may define the normalisation
 \[ \Tpoly(m) := \max_{(\lg m)^4 < \lg L < m} \; \frac{\Tpoly(m, L)}{m L \lg L}. \]

In our algorithm we must often perform auxiliary arithmetic operations on `small' integers.
These will always be handled via the Sch\"onhage--Strassen algorithm
\cite{SS-multiply} and Newton's method \cite[Ch.~9]{vzGG-compalg3};
thus we may compute products, quotients and remainders of $n$-bit integers
in $O(n \lg n \lg \lg n)$ bit operations.

\begin{prop}
\label{prop:main-bound}
There exist absolute constants $m_0 > 2^{17}$, $C_1 > 0$ and $C_2 > 0$ with the following property.
For any admissible $m > m_0$, there exists an admissible $m' < (\log m)^4$ such that
\begin{equation}
\label{eq:main-bound}
 \Tpoly(m) < \left(4 + \frac{C_1}{\lg \lg m}\right) \Tpoly(m') + C_2.
\end{equation}
\end{prop}

\begin{proof}
Assume that we are given as input an admissible size $m = k(\lg k)^3$, the corresponding prime $p = p_0(m) = a \cdot 2^m + 1$, a transform length $L = 2^{\lg L}$ satisfying~\eqref{eq:lgL-interval}, a primitive $L$-th root of unity $\zeta \in \FF_p$, and a polynomial $F \in \FF_p[X]/(X^L - 1)$; our goal is to compute~$\hat F$.
For the base case $m \leq m_0$, we may compute~$\hat F$ using any convenient algorithm.
In what follows, we assume that $m > m_0$ and that $m_0$ is increased whenever necessary to accommodate statements that hold only for large~$m$.

Let us now detail the reductions (i)--(v) mentioned above.

\medskip
\step{reduce to short DFTs.}
In this step we reduce the given transform of length $L$ to a collection of short transforms of length
 \[ S := 2^{(\lg m)^2}. \]
By \eqref{eq:lgL-interval} we have $\lg L > (\lg m)^2$, so $S \divides L$.
Let $\omega := \zeta^{L/S}$; then $\omega$ is a primitive $S$-th root of unity in $\FF_p$.

Let $d := \lfloor \lg L / (\lg m)^2 \rfloor$, so that $\lg L = (\lg m)^2 d + d'$ where $0 \leq d' < (\lg m)^2$.
Applying the Cooley--Tukey method \cite{CT-fft} to the factorisation $L = S^d 2^{d'}$, the given transform of length $L$ may be decomposed into~$d$ layers, each consisting of $L/S$ transforms of length~$S$ (with respect to $\omega$), followed by $d'$ layers, each consisting of $L/2$ transforms of length~$2$.
Between each of these layers, we must perform $O(L)$ multiplications by `twiddle factors' in $\FF_p$, which are given by certain powers of $\zeta$.
(For further details of the Cooley--Tukey decomposition, see for example \cite[\S2.3]{HvdHL-mul}.)

Let us estimate the total cost of the twiddle factor multiplications.
We have $d \leq \lg L / (\lg m)^2$, and by \eqref{eq:lgL-interval} also $d' < (\lg m)^2 < \lg L / (\lg m)^2$, so the total number of twiddle factor multiplications is $O((d + d') L) = O(L \lg L / (\lg m)^2)$.
Using the Sch\"onhage--Strassen algorithm, each multiplication in $\FF_p$ costs at most $O(\lg p \lg \lg p \lg \lg \lg p)$ bit operations.
As $p = a \cdot 2^m + 1$ we have $\lg p = m + O(\lg a)$, so \eqref{eq:a-bound} implies that $\lg p = m + O(\lg m) = O(m)$.
Thus the cost of each multiplication in $\FF_p$ is $O(m (\lg m)^2)$ bit operations, and the total cost of the twiddle factor multiplications is $O(m L \lg L)$ bit operations.
This bound also covers the cost of the length~$2$ transforms (`butterflies'), each of which requires one addition and one subtraction in $\FF_p$.

In the Turing model, we must also account for the cost of rearranging data so that the inputs for each layer of short DFTs are stored sequentially on tape.
Using a fast matrix transpose algorithm, the cost per layer is $O(L \lg p \lg S) = O(Lm (\lg m)^2)$ bit operations (see \cite[\S2.3]{HvdHL-mul} for further details), so $O(m L \lg L)$ bit operations altogether.

Let $\Tshort(m, L)$ denote the number of bit operations required to perform $L/S$ transforms of length $S$ with respect to $\omega$, i.e., the cost of one layer of short transforms.
Since the number of layers of short transforms is $d \leq \lg L / (\lg m)^2$, the above discussion shows that
\begin{equation}
\label{eq:Tpoly-bound1}
 \Tpoly(m, L) < \frac{\lg L}{(\lg m)^2} \Tshort(m, L) + O(m L \lg L).
\end{equation}

\medskip
\step{reduce to short convolutions.}
In this step we use Bluestein's algorithm~\cite{Blu-dft} to convert the short transforms into convolution problems.
Suppose that at some layer of the main DFT we are given as input the short polynomials
 \[ a_t(X) = \sum_{i=0}^{S-1} a_{t,i} X^i \in \FF_p[X]/(X^S - 1), \qquad t = 1, \ldots, L/S. \]
We wish to compute $\hat a_1, \ldots, \hat a_{L/S}$, the DFTs with respect to $\omega$.

Let $\eta := \zeta^{L/2S}$ so that $\eta^2 = \omega$.
For $i = 0, \ldots, S-1$ and $t = 1, \ldots, L/S$, define
 \[ f_{t,i} := \eta^{i^2} a_{t,i}, \qquad g_i := \eta^{-i^2}, \]
and put
 \[ f_t(X) := \sum_{i=0}^{S-1} f_{t,i} X^i, \qquad g(X) := \sum_{i=0}^{S-1} g_i X^i, \]
regarded as polynomials in $\FF_p[X]/(X^S - 1)$.
We may compute all of the $g_i$, and compute all of the $f_{t,i}$ from the $a_{t,i}$, using $O((L/S) S) = O(L)$ operations in $\FF_p$.
Then one finds (see for example \cite[\S2.5]{HvdHL-mul}) that $(\hat a_t)_i = \eta^{i^2} h_{t,i}$, where
\begin{equation}
\label{eq:ht}
 h_t := f_t g = \sum_{i=0}^{S-1} h_{t,i} X^i \in \FF_p[X]/(X^S - 1).
\end{equation}
In other words, computing the short DFTs reduces to computing the products $f_1 g, \ldots, f_{L/S} g$ in $\FF_p[X]/(X^S - 1)$, plus an additional $O(L)$ operations in $\FF_p$.

Let $\Cshort(m, L)$ denote the cost of computing the products $f_1 g, \ldots, f_{L/S} g$ in $\FF_p[X]/(X^S - 1)$.
As noted in step (i), each multiplication in $\FF_p$ costs $O(m (\lg m)^2)$ bit operations, so the above discussion shows that
 \[ \Tshort(m, L) < \Cshort(m, L) + O(L m (\lg m)^2). \]
Substituting into \eqref{eq:Tpoly-bound1} yields
\begin{equation}
\label{eq:Tpoly-bound2}
 \Tpoly(m, L) < \frac{\lg L}{(\lg m)^2} \Cshort(m, L) + O(m L \lg L).
\end{equation}

\medskip
\step{reduce to bivariate product over $\ZZ$.}
In this step we transport the problem of computing the products $f_1 g, \ldots, f_{L/S} g$ in $\FF_p[X]/(X^S - 1)$ to the ring
 \[ \Rring := \ZZ[X,Y]/(X^S - 1, Y^k + a), \]
by cutting up each coefficient in $\FF_p$ into $k$ chunks of bit size
 \[ r := m/k = (\lg k)^3. \]
We note for future reference the estimate
\begin{equation}
\label{eq:r-bound}
 r = \left(1 + \frac{O(1)}{\lg \lg m} \right) (\lg m)^3;
\end{equation}
this follows from $m = k (\lg k)^3$, because
 \[ \lg m = \lg k + O(\lg \lg k) = \left(1 + \frac{O(\lg \lg k)}{\lg k}\right) \lg k = \left(1 + \frac{O(1)}{\lg \lg m} \right) \lg k. \]

Interpreting each $f_{t,i}$ and $g_i$ as an integer in the interval $[0, p)$, and decomposing them in base $2^r$, we write
\begin{equation}
\label{eq:ftij-defn}
 f_{t,i} = \sum_{j=0}^{k-1} f_{t,i,j} 2^{(k-1-j)r}, \qquad g_i = \sum_{j=0}^{k-1} g_{i,j} 2^{(k-1-j)r},
\end{equation}
where $f_{t,i,j}$ and $g_{i,j}$ are integers in the interval
\begin{equation}
\label{eq:fg-bound}
 0 \leq f_{t,i,j}, g_{i,j} \leq 2^r a.
\end{equation}
(In fact, they are less than $2^r$ for $j = 1, \ldots, k-1$; the bound $2^r a$ is only needed for the first term $j = 0$.)
Then define polynomials
 \[ F_t := \sum_{i=0}^{S-1} \sum_{j=0}^{k-1} f_{t,i,j} X^i Y^j, \qquad G := \sum_{i=0}^{S-1} \sum_{j=0}^{k-1} g_{i,j} X^i Y^j, \]
regarded as elements of $\Rring = \ZZ[X,Y]/(X^S - 1, Y^k + a)$, and let
 \[ H_t := F_t G = \sum_{i=0}^{S-1} \sum_{j=0}^{k-1} h_{t,i,j} X^i Y^j \]
be the corresponding products in $\Rring$ for $t = 1, \ldots, L/S$.

We claim that knowledge of $H_t$ determines $h_t$; specifically, that
\begin{equation}
\label{eq:H-overlap}
 h_{t,i} = \sum_{j=0}^{k-1} h_{t,i,j} 2^{(2k-2-j)r} \pmod p
\end{equation}
for each pair $(t,i)$.
To prove this, observe that by definition of multiplication in $\Rring$,
\begin{equation}
\label{eq:htij-formula}
 h_{t,i,j} = \sum_{i_1 + i_2 = i \bmod S} \Bigg(\sum_{\substack{j_1+j_2=j\bmod k \\j_1+j_2 < k}} f_{t,i_1,j_1} g_{i_2,j_2} - \sum_{\substack{j_1+j_2=j\bmod k \\j_1+j_2 \geq k}} a f_{t,i_1,j_1} g_{i_2,j_2} \Bigg).
\end{equation}
On the other hand, from \eqref{eq:ht} and \eqref{eq:ftij-defn} we have
\begin{align*}
 h_{t,i} & = \sum_{i_1+i_2=i \bmod S}  f_{t,i_1} g_{i_2} \\
         & = 2^{(2k-2)r} \sum_{i_1+i_2=i \bmod S}\  \sum_{j_1=0}^{k-1} \sum_{j_2=0}^{k-1} f_{t,i_1,j_1} g_{i_2,j_2} 2^{-(j_1+j_2)r},
\end{align*} 
and since $2^{-kr} = 2^{-m} = -a \pmod p$, we obtain
\begin{multline*}
   h_{t,i} = 2^{(2k-2)r} \sum_{i_1+i_2=i \bmod S}  \sum_{j=0}^{k-1} \\
      \phantom{==} \Bigg(\  \sum_{\substack{j_1+j_2=j\bmod k \\j_1+j_2 < k}} f_{t,i_1,j_1} g_{i_2,j_2} - \sum_{\substack{j_1+j_2=j\bmod k \\j_1+j_2 \geq k}} a f_{t,i_1,j_1} g_{i_2,j_2} \Bigg) 2^{-jr} \pmod p.
\end{multline*}
Comparing this expression with \eqref{eq:htij-formula} yields \eqref{eq:H-overlap}.
(The reason this works is that $f_{t,i}$ and $g_i$ are the images of $F_{t,i}$ and $G_i$ under the ring homomorphism from $\Rring$ to $\FF_p[X]/(X^S - 1)$ that sends $Y$ to $2^{-r}$, modulo a scaling factor of $2^{(k-1)r}$.)

Let us estimate the cost of using \eqref{eq:H-overlap} to compute $h_{t,i}$ for a single pair $(t,i)$, assuming that the $h_{t,i,j}$ are known.
By \eqref{eq:htij-formula} and \eqref{eq:fg-bound} we have
 \[ |h_{t,i,j}| \leq (Sk) a (2^r a)^2 = 2^{2r} S k a^3. \]
Then, using \eqref{eq:a-bound} and the fact that $k \leq m$, we obtain $|h_{t,i,j}| \leq 2^{2r} S C^3 m^7$, so
 \[ \lg |h_{t,i,j}| < 2r + (\lg m)^2 + 7 \lg m + O(1). \]
Moreover, as $m = k(\lg k)^3$, for large $m$ we certainly have $\lg k \leq \lg m - 1$, so
 \[ r = (\lg k)^3 \leq (\lg m - 1)^3 \leq (\lg m)^3 - 2(\lg m)^2, \]
and we deduce that
\begin{equation}
\label{eq:lg-htij-bound}
 \lg |h_{t,i,j}| < 2(\lg m)^3.
\end{equation}
In particular, by \eqref{eq:r-bound}, we see that $h_{t,i,j}$ has bit size $O((\lg m)^3) = O(r)$.
Now, to compute $h_{t,i}$, we first evaluate the sum in \eqref{eq:H-overlap} (in $\ZZ$) using a straightforward overlap-add procedure; this costs $O(kr) = O(m)$ bit operations.
Then we reduce the result modulo $p$; this costs a further $O(m (\lg m)^2)$ bit operations (using the Sch\"onhage--Strassen algorithm, as in step (i)).
The total cost over all pairs $(t, i)$ is $O((L/S)S m (\lg m)^2) = O(L m (\lg m)^2)$ bit operations.

Let $\Cbiv(m, L)$ denote the cost of computing the products $F_1 G, \ldots, F_{L/S} G$ in $\Rring$.
The above discussion shows that
 \[ \Cshort(m, L) < \Cbiv(m, L) + O(L m (\lg m)^2). \]
Substituting into \eqref{eq:Tpoly-bound2} yields
\begin{equation}
\label{eq:Tpoly-bound3}
  \Tpoly(m, L) < \frac{\lg L}{(\lg m)^2} \Cbiv(m, L) + O(m L \lg L).
\end{equation}

\medskip
\step{reduce to bivariate multiplication over $\FF_{p'}$.}
In this step we transfer the above multiplication problems from $\Rring = \ZZ[X,Y]/(X^S - 1, Y^k + a)$ to the ring
 \[ \Sring := \FF_{p'}[X,Y]/(X^S - 1, Y^k + a), \]
where $p' := p_0(m') = a' \cdot 2^{m'} + 1$ for a suitable choice of admissible size $m'$.
The idea is to choose $m'$ to be just large enough that computing the desired products modulo~$p'$ determines their coefficients unambiguously in $\ZZ$.

To achieve this, we will take $m' := k' (\lg k')^3$ where
 \[ k' := \left \lceil \frac{\beta}{(\lg \beta - 3 \lg \lg \beta)^3} \right \rceil, \qquad  \beta := 2(\lg m)^3. \]
Note that $m'$ is admissible (we may ensure that $m' > 2^{17}$ by taking $m$ sufficiently large).
We claim that with this choice of $m'$ we have
\begin{equation}
\label{eq:mprime-interval}
 \beta \leq m' < \left(1 + \frac{O(1)}{\lg \lg m}\right) \beta,
\end{equation}
and consequently, taking \eqref{eq:r-bound} into account,
\begin{equation}
\label{eq:mprime-estimate}
 m' = \left(2 + \frac{O(1)}{\lg \lg m} \right) r.
\end{equation}
To establish the first inequality in \eqref{eq:mprime-interval}, observe that since $k' \geq \beta / (\lg \beta)^3$, we have
 \[ \log_2 k' \geq \log_2 \beta - 3 \log_2 \lg \beta \geq \log_2 \beta - 3 \lg \lg \beta. \]
Hence $\lg k' \geq \lg \beta - 3 \lg \lg \beta$, which implies that $m' = k' (\lg k')^3 \geq \beta$.  
For the second inequality, since $k' \leq \beta$ we have
\begin{align*}
  m' = k'(\lg k')^3 & \leq \left( \frac{\beta}{(\lg \beta - 3\lg \lg \beta)^3} + 1\right) (\lg \beta)^3 \\
                    & = \left(\frac{(\lg \beta)^3}{(\lg \beta - 3 \lg \lg \beta)^3} + \frac{(\lg \beta)^3}{\beta} \right) \beta \\
                    & = \left(1 + \frac{O((\lg \lg \beta)^3)}{(\lg \beta)^3} \right) \beta = \left(1 + \frac{O(1)}{\lg \lg m}\right) \beta.
\end{align*}

Let us estimate the cost of computing $p' = a' \cdot 2^{m'} + 1$.
Hypothesis P ensures that $a' < C(m')^2$, so we may locate $p'$ by testing each candidate $a' = 1, \ldots, C(m')^2$ using a naive primality test (trial division) in $2^{O(m')}$ bit operations.
By~\eqref{eq:lgL-interval} and \eqref{eq:mprime-interval} this amounts to
\begin{equation}
\label{eq:find-prime-bound}
 2^{O(m')} = 2^{O(\beta)} = 2^{O((\lg m)^3)} = 2^{O((\lg L)^{3/4})} = O(L)
\end{equation}
bit operations.

Now let $F_1, \ldots, F_{L/S}, G \in \Rring$ be as in step (iii), and let $u_1, \ldots, u_{L/S}, v$ be their images in $\Sring$; that is,
 \[ u_t := \sum_{i=0}^{S-1} \sum_{j=0}^{k-1} u_{t,i,j} X^i Y^j, \qquad v := \sum_{i=0}^{S-1} \sum_{j=0}^{k-1} v_{i,j} X^i Y^j, \]
where $u_{t,i,j}$ and $v_{i,j}$ are the images in $\FF_{p'}$ of $f_{t,i,j}$ and $g_{i,j}$.
Computing these images amounts to zero-padding each coefficient up to $\lg p'$ bits; by \eqref{eq:mprime-estimate} we have $\lg p' = O(m') = O(r)$, so the total cost of this step is $O((L/S)Skr) = O(Lm)$ bit operations.

Let $w_t := u_t v$ for each $t$.
Clearly $w_t$ is the image in $\Sring$ of $H_t = F_t G$.
By \eqref{eq:lg-htij-bound} and \eqref{eq:mprime-interval}, the coefficients $h_{t,i,j}$ of $H_t$ are completely determined by those of $w_t$, as $|h_{t,i,j}| \leq 2^{\beta-1}$ and $p' \geq 2^{m'} + 1 \geq 2^{\beta} + 1$.
Moreover, this lifting can be carried out in linear time, so the cost of deducing $H_t$ from $w_t$ (for all $t$) is again $O(Lm)$ bit operations.

Let $\Ctiny(m, L)$ denote the cost of computing the products $u_1 v, \ldots, u_{L/S} v$ in $\Sring$.
The above discussion shows that
 \[ \Cbiv(m, L) < \Ctiny(m, L) + O(L m). \]
Substituting into \eqref{eq:Tpoly-bound3} yields
\begin{equation}
\label{eq:Tpoly-bound4}
  \Tpoly(m, L) < \frac{\lg L}{(\lg m)^2} \Ctiny(m, L) + O(m L \lg L).
\end{equation}

\step{reduce to DFTs over $\FF_{p'}$.}
Since $2^{m'} \divides p' - 1$ and $\lg S = (\lg m)^2 < m'$, there exists a primitive $S$-th root of unity $\zeta' \in \FF_{p'}$.
We may find one such primitive root by a brute force search in $2^{O(m')} = O(L)$ bit operations (see \eqref{eq:find-prime-bound}).

We will compute the products $w_t = u_t v$ in $\Sring$ by first performing DFTs with respect to~$X$, and then multiplying pointwise in $\FF_{p'}[Y]/(Y^k + a)$.
More precisely, for each $j = 0, \ldots, k-1$ let
 \[ U_{t,j} := \sum_{i=0}^{S-1} u_{t,i,j} X^i, \qquad V_j := \sum_{i=0}^{S-1} v_{i,j} X^i, \]
regarded as polynomials in $\FF_{p'}[X]/(X^S - 1)$.
We call $\textsc{Transform}$ recursively to compute the transforms of each $U_{t,j}$ and $V_j$ with respect to $\zeta'$, i.e., to compute the polynomials
 \[ u_t((\zeta')^i, Y) = \sum_{j=0}^{k-1} U_{t,j}((\zeta')^i) Y^j, \qquad v((\zeta')^i, Y) = \sum_{j=0}^{k-1} V_j((\zeta')^i) Y^j, \]
as elements of $\FF_{p'}[Y]/(Y^k + a)$, for each $i = 0, \ldots, S-1$ and $t = 1, \ldots, L/S$.
The precondition corresponding to~\eqref{eq:lgL-interval} for these recursive calls is
 \[ (\lg m')^4 < \lg S < m'. \]
This is certainly satisfied for large $m$, as $\lg S = (\lg m)^2$ and $m' \sim 2(\lg m)^3$ by \eqref{eq:mprime-interval}.
There are $(L/S + 1)k$ transforms, so their total cost is $(L/S + 1)k \Tpoly(m', S)$ bit operations.

We next compute the pointwise products
 \[ w_t((\zeta')^i, Y) = u_t((\zeta')^i, Y) \cdot v((\zeta')^i, Y) \]
in $\FF_{p'}[Y]/(Y^k + a)$ for each $i$ and $t$.
Using the Sch\"onhage--Strassen algorithm (both the integer variant and the polynomial variant \cite{CK-fastmult}), the cost of each product is $O((k \lg k \lg \lg k) (\lg p' \lg \lg p' \lg \lg \lg p'))$ bit operations.
Using the bounds $\lg p' = O(m')$, $km' = O(kr) = O(m)$, and $m' = O((\lg m)^3)$, this amounts to
\begin{align*}
  O((k \lg m \lg \lg m) (m' \lg m' \lg \lg m')) & = O(m \lg m (\lg \lg m)^2 \lg \lg \lg m) \\
                                                & = O(m (\lg m)^2)
\end{align*}
bit operations, or $O(L m (\lg m)^2)$ bit operations over all $t$ and $i$.

Finally, we perform inverse DFTs with respect to $X$ to recover $w_1, \ldots, w_{L/S}$.
It is well known that these inverse DFTs may be computed by the same algorithm as the forward DFT, with $\zeta'$ replaced by $(\zeta')^{-1}$, followed by a division by $S$.
The divisions cost $O(L m (\lg m)^2)$ bit operations altogether (they are no more expensive than the pointwise multiplications), so the cost of this step is $(L/S) k \Tpoly(m', s) + O(L m (\lg m)^2)$.

The procedure just described requires some data rearrangement, so that the DFTs and pointwise multiplication steps can access the necessary data sequentially.
Using a fast matrix transpose algorithm, this costs $O(L k \lg p' \lg k) = O(L m \lg m)$ bit operations altogether.

Combining all the contributions mentioned above, we obtain
 \[ \Ctiny(m, L) < \left(\frac{2L}{S} + 1 \right) k \Tpoly(m', S) + O(L m (\lg m)^2). \]
Substituting into \eqref{eq:Tpoly-bound4} yields
 \[ \Tpoly(m, L) < \frac{\lg L}{(\lg m)^2} \left( \frac{2L}{S} + 1 \right) k \Tpoly(m', S) + O(m L \lg L). \]

{\sl Conclusion.}
This concludes the description of the algorithm; it remains to establish \eqref{eq:main-bound}.
Dividing the previous inequality by $m L \lg L$, and recalling that $S = 2^{(\lg m)^2}$, we obtain
 \[ \frac{\Tpoly(m, L)}{m L \lg L} < \left(2 + \frac{S}{L}\right) \frac{km'}{m} \cdot \frac{\Tpoly(m', S)}{m' S \lg S} + O(1). \]
By definition of $\Tpoly(m')$ we have $\Tpoly(m', S)/(m' S \lg S) \leq \Tpoly(m')$.
Equation \eqref{eq:mprime-estimate} implies that
 \[ km' = \left(2 + \frac{O(1)}{\lg \lg m}\right) kr = \left(2 + \frac{O(1)}{\lg \lg m}\right) m, \]
and \eqref{eq:lgL-interval} yields $S/L < 2^{(\lg m)^2 - (\lg m)^4} = O(1 / \lg \lg m)$.
Therefore
 \[ \frac{\Tpoly(m, L)}{m L \lg L} < \left(4 + \frac{O(1)}{\lg \lg m}\right) \Tpoly(m') + O(1). \]
Taking the maximum over all $\lg L$ yields the desired bound \eqref{eq:main-bound}.
\end{proof}
\begin{cor}
\label{cor:bound}
We have $\Tpoly(m) = O(4^{\log^* m})$ for admissible $m \to \infty$.
\end{cor}
The corollary could be deduced from Proposition \ref{prop:main-bound} by using \cite[Prop.~8]{HvdHL-mul}.
We give a simpler (but less general) argument here.
\begin{proof}
Let $m_0$, $C_1$ and $C_2$ be as in Proposition \ref{prop:main-bound}.
We may assume, increasing $m_0$ if necessary, that
 \[ (\log m)^{1/2} < \log(m^{1/8}) \qquad \text{and} \qquad \frac{C_1}{\lg \lg m} < 4^{-\log^*(m^{1/8})} \]
for all $m > m_0$.
Define
 \[ B := \max\big(C_2, \max_{\substack{\text{$m$ admissible} \\ m \leq m_0}} \Tpoly(m)\big). \]
We will prove that
\begin{equation}
\label{eq:Tpoly-bound}
 \Tpoly(m) < (4^{\log^*(m^{1/8}) + 1} - 1)B
\end{equation}
for all admissible $m$.

If $m \leq m_0$ then \eqref{eq:Tpoly-bound} holds by the definition of $B$, so we may assume that $m > m_0$.
By Proposition \ref{prop:main-bound}, there exists an admissible $m' < (\log m)^4$ such that
\begin{align*}
 \Tpoly(m) & < \left(4 + \frac{C_1}{\lg \lg m}\right) \Tpoly(m') + C_2 \\ 
           & < (4 + 4^{-\log^*(m^{1/8})}) \Tpoly(m') + B.
\end{align*}
Since $(m')^{1/8} < (\log m)^{1/2} < \log(m^{1/8})$, we have $\log^*((m')^{1/8}) \leq \log^*(m^{1/8}) - 1$.
By induction on $\log^*(m^{1/8})$, we obtain
\begin{align*}
 \Tpoly(m) & < (4 + 4^{-\log^*(m^{1/8})}) (4^{\log^*(m^{1/8})} - 1)B + B \\
           & = (4^{\log^*(m^{1/8}) + 1} - 4^{-\log^*(m^{1/8})} - 2) B \\
           & < (4^{\log^*(m^{1/8}) + 1} - 1) B.
\end{align*}
This establishes \eqref{eq:Tpoly-bound}, and the corollary follows immediately.
\end{proof}

\subsection{Application to integer multiplication}

We are now in a position to prove the main result.

\begin{proof}[Proof of Theorem \ref{thm:main}]
We are given as input two positive integers $u, v < 2^n$ for some large $n$; our goal is to compute $uv$.

Define
 \[ k := \left\lceil \frac{(5/2) \lg n}{(\lg \lg n)^3} \right\rceil, \qquad m := k(\lg k)^3. \]
We have
 \[ \lg k = \lg \lg n + O(\lg \lg \lg n) = \left(1 + o(1) \right) \lg \lg n, \]
so
 \[ 2 \lg n < m < 3 \lg n \]
for large $n$.
We may assume that $n$ is large enough so that $m > 2^{17}$; then $m$ is admissible.

Let $b := \lfloor m/4 \rfloor$ and $d := \lceil n/b \rceil$.
We decompose $u$ and $v$ in base $2^b$ as
 \[ u = \sum_{i=0}^{d-1} u_i 2^{bi}, \qquad v = \sum_{i=0}^{d-1} v_i 2^{bi}, \qquad 0 \leq u_i, v_i < 2^b, \]
and define polynomials
 \[ U(X) := \sum_{i=0}^{d-1} u_i X^i, \qquad V(X) := \sum_{i=0}^{d-1} v_i X^i \qquad \in \ZZ[X]. \]
Let $W := UV \in \ZZ[X]$.
The coefficients of $W$ have at most $2b + \lg d = O(m)$ bits, so the product $uv = W(2^b)$ may be recovered from $W(X)$ in $O(n)$ bit operations.

Let $p := p_0(m) = a \cdot 2^m + 1$.
We may find $p$ by testing each value of $a$ up to $Cm^2$ using a polynomial-time primality test \cite{AKS-primes}, in $m^{O(1)} = (\lg n)^{O(1)}$ bit operations.
Also put $\ell := \lg(10n/m)$ and $L := 2^\ell$.
The polynomial $W$ is determined by its image in $\FF_p[X]/(X^L - 1)$, as
 \[ d \leq \frac{n}{b} + 1 \leq \frac{n}{m/4 - 1} + 1 < \frac{5n}{m} \leq L/2 \]
and
 \[ 2b + \lg d \leq m/2 + \lg n < m \]
for large $n$.

To compute the product in $\FF_p[X]/(X^L - 1)$, we will use $\textsc{Transform}$ to perform DFTs and inverse DFTs, and multiply pointwise in $\FF_p$.
The precondition \eqref{eq:lgL-interval} certainly holds for large $n$.
According to \cite{Shp-primitive}, we may find a suitable primitive root in $\FF_p$ in
 \[ p^{1/4+o(1)} < (2^{m/4})^{1+o(1)} < (2^{(3/4)\lg n})^{1+o(1)} = n^{3/4 + o(1)} \]
bit operations.
By Corollary \ref{cor:bound}, we conclude that
\begin{align*}
 \Mint(n) & < 3 \Tpoly(m, L) + O(L m \lg m \lg \lg m) \\
          & < 3 m L \lg L \Tpoly(m) + O(L m \lg m \lg \lg m) \\
          & = O(n \lg n \, 4^{\log^* m}) + O(n \lg \lg n \lg \lg \lg n) \\
          & = O(n \lg n \, 4^{\log^* n}). \qedhere
\end{align*}
\end{proof}

\section*{Acknowledgments}

The authors thank an anonymous referee, whose comments helped to greatly improve the presentation of these results, and also Igor Shparlinski and Liangyi Zhao, for helpful discussions about bounds for primes in arithmetic progressions.

\bibliography{vanilla}

\providecommand{\bysame}{\leavevmode\hbox to3em{\hrulefill}\thinspace}
\providecommand{\MR}{\relax\ifhmode\unskip\space\fi MR }
\providecommand{\MRhref}[2]{%
  \href{http://www.ams.org/mathscinet-getitem?mr=#1}{#2}
}
\providecommand{\href}[2]{#2}
\begin{thebibliography}{10}

\bibitem{AKS-primes}
M.~Agrawal, N.~Kayal, and N.~Saxena, \emph{P{RIMES} is in {P}}, Ann. of Math.
  (2) \textbf{160} (2004), no.~2, 781--793. \MR{2123939}

\bibitem{Blu-dft}
L.~I. Bluestein, \emph{A linear filtering approach to the computation of
  discrete {F}ourier transform}, IEEE Transactions on Audio and
  Electroacoustics \textbf{18} (1970), no.~4, 451--455.

\bibitem{CK-fastmult}
D.~G. Cantor and E.~Kaltofen, \emph{On fast multiplication of polynomials over
  arbitrary algebras}, Acta Inform. \textbf{28} (1991), no.~7, 693--701.
  \MR{1129288 (92i:68068)}

\bibitem{Cha-charsums}
M.~Chang, \emph{Short character sums for composite moduli}, J. Anal. Math.
  \textbf{123} (2014), 1--33. \MR{3233573}

\bibitem{CT-fft}
J.~W. Cooley and J.~W. Tukey, \emph{An algorithm for the machine calculation of
  complex {F}ourier series}, Math. Comp. \textbf{19} (1965), 297--301.
  \MR{0178586}

\bibitem{CT-zmult}
S.~Covanov and E.~Thom\'e, \emph{Fast integer multiplication using generalized
  {F}ermat primes}, preprint \url{http://arxiv.org/abs/1502.02800}, 2016.

\bibitem{CF-DWT}
R.~Crandall and B.~Fagin, \emph{Discrete weighted transforms and large-integer
  arithmetic}, Math. Comp. \textbf{62} (1994), no.~205, 305--324. \MR{1185244}

\bibitem{Fur-faster1}
M.~F{\"u}rer, \emph{Faster integer multiplication}, S{TOC}'07---{P}roceedings
  of the 39th {A}nnual {ACM} {S}ymposium on {T}heory of {C}omputing, ACM, New
  York, 2007, pp.~57--66. \MR{2402428 (2009e:68124)}

\bibitem{Fur-faster2}
\bysame, \emph{Faster integer multiplication}, SIAM J. Comput. \textbf{39}
  (2009), no.~3, 979--1005. \MR{2538847 (2011b:68296)}

\bibitem{gmp}
Torbj\"{o}rn Granlund and the {GMP}~development team, \emph{{GNU} {M}ultiple
  {P}recision {A}rithmetic {L}ibrary}, version 6.1.2, available at
  \url{http://gmplib.org/}, 2017.

\bibitem{GP-least}
A.~Granville and C.~Pomerance, \emph{On the least prime in certain arithmetic
  progressions}, J. London Math. Soc. (2) \textbf{41} (1990), no.~2, 193--200.
  \MR{1067261}

\bibitem{Har14}
D.~Harvey, \emph{Faster arithmetic for number-theoretic transforms}, J.
  Symbolic Comput. \textbf{60} (2014), 113--119.

\bibitem{Har-truncmul}
\bysame, \emph{Faster truncated integer multiplication},
  \url{https://arxiv.org/abs/1703.00640}, 2017.

\bibitem{HvdHL-mul}
D.~Harvey, J.~van~der Hoeven, and G.~Lecerf, \emph{Even faster integer
  multiplication}, J. Complexity \textbf{36} (2016), 1--30. \MR{3530637}

\bibitem{vdH:f2kmul}
\bysame, \emph{Fast polynomial multiplication over $\mathbb{F}_{2^{60}}$},
  Proc. ISSAC '16 (New York, NY, USA), ACM, 2016, pp.~255--262.

\bibitem{vdH:ffmul}
\bysame, \emph{Faster polynomial multiplication over finite fields}, J. ACM
  \textbf{63} (2017), no.~6, Article~52.

\bibitem{HB-almost}
D.~R. Heath-Brown, \emph{Almost-primes in arithmetic progressions and short
  intervals}, Math. Proc. Cambridge Philos. Soc. \textbf{83} (1978), no.~3,
  357--375. \MR{0491558 (58 \#10789)}

\bibitem{HB-siegel}
\bysame, \emph{Siegel zeros and the least prime in an arithmetic progression},
  Quart. J. Math. Oxford Ser. (2) \textbf{41} (1990), no.~164, 405--418.
  \MR{1081103}

\bibitem{HB-linnik}
\bysame, \emph{Zero-free regions for {D}irichlet {$L$}-functions, and the least
  prime in an arithmetic progression}, Proc. London Math. Soc. (3) \textbf{64}
  (1992), no.~2, 265--338. \MR{1143227 (93a:11075)}

\bibitem{vdH:ff2mul}
J.~van~der Hoeven, R.~Larrieu, and G.~Lecerf, \emph{Implementing fast carryless
  multiplication}, Tech. report, HAL, 2017,
  \verb|http://hal.archives-ouvertes.fr/hal-01579863|.

\bibitem{LPS-leastprime}
J.~Li, K.~Pratt, and G.~Shakan, \emph{A lower bound for the least prime in an
  arithmetic progression}, preprint \url{https://arxiv.org/abs/1607.02543},
  2016.

\bibitem{Pap-complexity}
C.~H. Papadimitriou, \emph{Computational complexity}, Addison-Wesley Publishing
  Company, Reading, MA, 1994. \MR{1251285 (95f:68082)}

\bibitem{Pol71}
J.~M. Pollard, \emph{The fast {Fourier} transform in a finite field},
  Mathematics of Computation \textbf{25} (1971), no.~114, 365--374.

\bibitem{SS-multiply}
A.~Sch{\"o}nhage and V.~Strassen, \emph{Schnelle {M}ultiplikation grosser
  {Z}ahlen}, Computing (Arch. Elektron. Rechnen) \textbf{7} (1971), 281--292.
  \MR{0292344 (45 \#1431)}

\bibitem{ntl-9.9.1}
V.~Shoup, \emph{{NTL}: a library for doing number theory ({V}ersion 9.9.1)},
  \url{http://www.shoup.net/ntl/}.

\bibitem{Shp-primitive}
I.~Shparlinski, \emph{On finding primitive roots in finite fields}, Theoret.
  Comput. Sci. \textbf{157} (1996), no.~2, 273--275. \MR{1389773}

\bibitem{M74207281}
GIMPS team, \emph{Mersenne {P}rime {N}umber discovery - $2^{74207281}-1$ is
  {P}rime}, \url{http://www.mersenne.org/primes/?press=M74207281}, 2016.

\bibitem{vzGG-compalg3}
J.~von~zur Gathen and J.~Gerhard, \emph{Modern computer algebra}, third ed.,
  Cambridge University Press, Cambridge, 2013. \MR{3087522}

\bibitem{Wag-least}
S.~S. Wagstaff, Jr., \emph{Greatest of the least primes in arithmetic
  progressions having a given modulus}, Math. Comp. \textbf{33} (1979),
  no.~147, 1073--1080. \MR{528061}

\bibitem{Xyl-linnik}
T.~Xylouris, \emph{On the least prime in an arithmetic progression and
  estimates for the zeros of {D}irichlet {$L$}-functions}, Acta Arith.
  \textbf{150} (2011), no.~1, 65--91. \MR{2825574 (2012m:11129)}

\end{thebibliography}

\end{document}